\documentclass[12pt]{amsart}
\usepackage{amssymb,latexsym,amsmath,amscd,amsthm}
\usepackage[latin1]{inputenc}
\usepackage{epsfig}
\usepackage{psfrag}
\usepackage{enumerate}
\usepackage{mathtools}
\newcommand{\SSM}{\mathtt{SSM}}
\newcommand{\Ka}{\mathtt{K81}}
\newcommand{\Kb}{\mathtt{K80}}
\newcommand{\JC}{\mathtt{JC69}}
\newcommand{\B}{\mathcal{B}}

\newcommand{\PP}{\mathbb{P}}

\newcommand{\CC}{\mathbb{C}}
\newcommand{\cc}{\mathcal{C}}

\newcommand{\LL}{\mathcal{L}}
\newcommand{\MM}{\mathcal{M}}

\newcommand{\DD}{\mathcal{D}}

\newcommand{\lra}{\longrightarrow}
\newcommand{\ra}{\rightarrow}

\newcommand{\X}{\mathtt{X}}
\newcommand{\x}{\mathtt{x}}
\newcommand{\y}{\mathtt{y}}
\newcommand{\z}{\mathtt{z}}
\renewcommand{\a}{\mathtt{A}}
\renewcommand{\c}{\mathtt{C}}
\newcommand{\g}{\mathtt{G}}
\renewcommand{\t}{\mathtt{T}}

\newcommand{\im}{\rm Im}

\newcommand{\EE}{\mathbb{E}}

\theoremstyle{definition}
\newtheorem{defi}{Definition}
\newtheorem{ex}[defi]{Example}
\theoremstyle{plain}
\newtheorem{lema}[defi]{Lemma}
\newtheorem{thm}[defi]{Theorem}

\newtheorem{prop}[defi]{Proposition}
\newtheorem{cor}[defi]{Corollary}
\newtheorem{rk}[defi]{Remark}

\newtheorem{notat}[defi]{Notation}
\newtheorem{teo-def}[defi]{Theorem/Definition}
\newenvironment{proofof}[1]{\noindent {\textit{Proof of {#1}.}}}{$\square$ \vspace{3mm}}
\DeclareMathOperator{\rep}{Par}

\title{The space of phylogenetic mixtures for equivariant models}
\author{Marta Casanellas}
\address{Departament de Matemàtica Aplicada I. ETSEIB. Universitat Polit\`ecnica de Catalunya. Avinguda Diagonal 647. 08028 Barcelona. Spain.}
\email{marta.casanellas@upc.edu}
\author{Jesús  Fernández-Sánchez}
\address{Departament de Matemàtica Aplicada I. ETSEIB. Universitat Polit\`ecnica de Catalunya. Avinguda Diagonal 647. 08028 Barcelona. Spain.}
\author{Anna M. Kedzierska}
\address{Centre for Genomic Regulation (CRG). Dr. Aiguader 88. 08003 Barcelona. Spain.}
\thanks{Research of the first and second authors partially supported by Ministerio de Educaci\'on y Ciencia MTM2009-14163-C02-02.}
\begin{document}

\begin{abstract}
\textbf{Background:}

The selection of an evolutionary model to best fit given
  molecular data is usually a heuristic choice. In his seminal book,
  J. Felsenstein suggested that certain linear equations satisfied by
  the expected probabilities of patterns observed at the leaves of a
  phylogenetic tree could be used for model selection. It remained an
  open question, however, whether these equations were sufficient to
  fully characterize the evolutionary model under consideration.

\textbf{Results:} Here we prove that, for most equivariant models of
evolution, the space of distributions satisfying these linear
equations coincides with the space of distributions arising from
 mixtures of trees. In other words, we prove that the evolution
of an observed
multiple sequence  alignment can be modeled by  a mixture of
phylogenetic trees under an equivariant evolutionary model if and only
if the distribution of patterns at its columns satisfies the linear
equations mentioned above.  Moreover, we provide a set of linearly
independent equations defining this space of phylogenetic mixtures for each equivariant model and
for any number of taxa. Lastly, we use these results to
perform a study of identifiability of phylogenetic mixtures.

\textbf{Conclusions:} The space of phylogenetic mixtures under
equivariant models is a linear space that fully characterizes
the evolutionary model. We provide an explicit algorithm to obtain the
equations defining these spaces for a number of models and taxa. Its implementation has proved to be a powerful tool for model selection.

\end{abstract}

\maketitle

%
%
\section*{Background}

The principal goal of phylogenetics is to reconstruct the ancestral
relationships among organisms. Most popular phylogenetic reconstruction methods are based on mathematical
models describing the molecular evolution of DNA.  In spite of this, there
exists no unified framework for model selection and the results are
highly dependent on the models and methods used in the analysis (cf. \cite{Posada2001}).

In this paper we assume the Darwinian model of evolution proceeding along phylogenetic trees and address the following question: how can
the data evolving under a particular model be
characterized? In other words, we look for invariants of the
DNA patterns which have evolved following a tree (or a mixture of trees, as we will see
below) under a particular model. The answer to this question
provided in this paper leads to a complete characterization of the evolutionary model and to a novel model selection tool, which is valid for any mixture of trees.

In what follows, we briefly explain the motivation for this
  work. It has been shown that  if the evolution along a phylogenetic
  tree is described  by a particular model, the expected probabilities of nucleotide patterns
at the leaves of the tree satisfy certain
equalities (see e.g. \cite[p.375]{Felsenstein2003}). Several authors (e.g. \cite{Felsenstein2003}, \cite{FuLi}, \cite{Steel1992}) pointed out that
these equalities could
potentially be used to test the fitness of the model of base change. The
full set of equations required  for viable model selection, however, was unknown. The objective of this work is to fill in this gap and
to go a step further into practical aplication by providing an
algorithm to compute the required invariants for model selection.


In this work we consider a group of equivariant
models (\cite{Draisma}, \cite{CFS3}). These models
are Markov processes on trees, whose transition matrices satisfy
certain symmetries: the Jukes-Cantor model, the Kimura 2 and 3
parameter models, the strand symmetric model, and the general Markov model. Our
first important result, Theorem~\ref{caractLG}, states that if evolution occurs
according to trees (or even mixtures of trees) under these equivariant models, then the model of
evolution is completely determined by the linear space defined by the aforementioned equalities. By exhaustively
studying the group of symmetries of these models, we also give a
straightforward combinatorial way of determining the equations of this
linear space (see Theorem \ref{teo_indep_equations}). The
implementation of the algorithm producing the equations is available
as a package \texttt{SPIn} (\cite{KDGC},
http://genome.crg.es/cgi-bin/phylo$\_$mod$\_$sel/AlgModelSelection.pl.),
which has proved to be a successful tool in evolutionary model selection.

Our main technique consists in proving that the linear space above
coincides with the space $\DD^{\MM}$ of \textit{phylogenetic
mixtures} evolving under the model $\MM$, i.e. the set of
points that are linear combinations of points lying in the
phylogenetic varieties $CV^{\MM}_T$ (see Preliminaries section for  specific definitions). In biological words and in the stochastic context, this is the
set of vectors of expected pattern frequencies for mixtures of trees evolving under the model
$\MM$ (not necessarily the same tree topology in the mixture, and not necessarily the same transition matrices when the tree topologies coincide). In phylogenetics, the so-called i.i.d. hypothesis  (independent and identically distributed) about the sites of an alignment is prevalent
in the most simple models. When the assumption
``identically distributed'' is replaced it by ``distributed
according to the same evolutionary model'', one obtains a
phylogenetic mixture.

Phylogenetic mixtures are useful in modeling heterogeneous
evolutionary processes, e.g. data comprising multiple genes, selected
codon positions, or rate variation across sites (e.g. \cite{Semple2003}).  Among a plethora of applications, they are used in orthology predictions,
gene and genome annotations, species tree reconstructions, and drug target
identifications.

In addition to the main result, we determine the dimension of these
linear spaces and use it to give an upper bound, $h_0(n)$, on the number of
mixtures that should be used in phylogenetic reconstruction on $n$
taxa.
This relates to the so-called \emph{identifiability} problem in phylogenetic
mixtures, which can be posed as determining the conditions that guarantee that the model
parameters (discrete parameters in the form of tree topologies and the
continuous parameters of the root and model distributions) can be recovered from
the data. Identifiability is crucial for consistency of the maximum
likelihood approaches and, though extensively studied in the
phylogenetic context, few results are known (see for instance \cite{Allman2006},
\cite{APRS}, \cite{Stefanovic},
\cite{RhodesSullivant},\cite{ChaiHousworth}).

In brief, in Theorem \ref{noident} we prove that either the tree topologies or
the continuous parameters are not generically identifiable for
mixtures on more than $h_0(n)$ trees under equivariant models. Here
$h_0(n)$ is the quotient of the dimension of the linear space
$\DD^{\MM}$ (computed in Proposition \ref{dim_models}) by the number
of free parameters of $\MM$ on a trivalent tree plus one. For example,
for four taxa and the Jukes-Cantor model (resp. the Kimura 3-parameter model)
this result proves that mixtures on three (resp. four) or more  taxa are not
identifiable (i.e. either the discrete or the continuous parameters cannot be fully identified).
A detailed discussion on this subject is provided in the last section.

The main tools used in this work are algebraic geometry and group
theory. The reader is referred to \cite{Harris} and \cite{Serre}
for general references on these topics.

\section*{Main text}

\subsection*{Preliminaries}

Phylogenetic trees and Markov models of evolution have been widely
used in the literature. In what follows  we fix the  notation needed to deal with them in our setting.

Let $n$ be a positive integer and denote by $[n]$ the set $\{1, 2, \dots, n\}$.
A \emph{phylogenetic tree} $T$ on the set of taxa $[n]$ is a tree
(i.e. a connected graph with no loops),  whose $n$ leaves are bijectively labeled by $[n]$. Its vertices represent species or other biological entities and its edges represent evolutionary processes between the vertices.

We allow internal vertices of any
degree and if all the internal vertices are of degree 3 we say that the tree is
\emph{trivalent}.
We will denote the set of vertices of $T$ by $N(T)$, the set of edges by $E(T)$, and the set of interior nodes by $Int(T)$.
A \textit{rooted tree} is a tree together with a distinguished node $r$ called the \emph{root}. The root induces an orientation on the edges of $T$, whereby the root represents the common ancestor to all the species represented in the tree. If $e$ is an edge of a rooted tree $T$, we write $pa(e)$ and $ch(e)$ for its parent vertex (origin) and its child  vertex (end), respectively.
Two unrooted phylogenetic trees on the set of taxa $[n]$ are said to
have the \emph{same tree topology} if their labeled graphs have the
same topology.

We fix a positive integer $k$ and an ordered set $B=\{b_1,b_2,\ldots,b_k\}.$ For example, for most applications we take $B=\{\texttt{A},\texttt{C},\texttt{G},\texttt{T}\}$ to be the set of nucleotides in a DNA sequence. We may think of $B$ as the set of states of a discrete random variable.  We call $W$ the complex vector space $W=\langle B
\rangle_{\mathbb{C}}$ spanned by $B$, so that $B$ is a natural basis of $W$. %
For algebraic convenience, we usually work over the complex field and restrict to the stochastic setting when necessary.
Vectors in $W$ are thought of as probability distributions
on the set of states $B$ if their coordinates are non-negative and sum
to one. In this setting   the vector
$\sum c_i b_i$ means that observation $b_i$ occurs with probability $c_i$.
 From now on, we will identify vectors in $W$ with their coordinates in the basis
$B$ written as a column vector, e.g.  we identify $\sum_k b_k$ with the vector $\mathbf{1}=(1,1,\ldots,1)^t\in W$.

In order to model molecular evolution on a phylogenetic tree $T$, we  consider a Markov
process specified by a root distribution, $\pi \in W$,  and a collection of transition matrices,  $\mathbf{A}=\left(A^{e}\right)_{e \in E(T)}$, where each
$A^{e}$ is a $k\times k$-matrix in $End(W).$ The matrices $A^{e}$
represent the conditional probabilities of substitution between the
states in $B$ from the parent node $pa(e)$ to the child node
$ch(e)$ of $e$. We adopt  the convention that the  matrices
$A^{e}$ act on $W$ from the right, i.e. a vector $\omega^t$ in $pa(e)$ maps to $\omega^t A^{e}$ in $ch(e)$.

%
%


Distinct forms of the transition matrices give rise to different
evolutionary models. Using the terminology introduced above, we
proceed to the definition of evolutionary models
used throughout this work.

\begin{defi}\label{em}
An \emph{(algebraic) evolutionary model $\mathcal{M}$} is
specified by giving a vector subspace $W_0\subset W$ such that $\mathbf{1}^t  \pi \neq 0$ for some $\pi$ in $W_0$, together with a
multiplicatively closed vector subspace $Mod$  (for \emph{model}) of $M_k(\CC)$ containing the identity matrix.
We will usually denote such a model by $\MM=(W_0, Mod)$.
We define the  \emph{stochastic evolutionary model $s\MM=(sW_0,sMod)$ associated to $\MM$} by taking  $sW_0=\{\pi\in W_0 : \mathbf{1}^t \pi =1\}$ and $sMod=\{A\in Mod: A\mathbf{1}=\mathbf{1}\}$.
The term ``stochastic'' refers to the fact that, by restricting
 to the points in the spaces with non-negative real
entries, we obtain distributions and Markov matrices.
A phylogenetic tree $T$ together with the parameters  $\pi$ and $\mathbf{A}=\left(A^{e}\right)_{e \in E(T)}$ is said to \emph{evolve under the algebraic evolutionary model} $\MM$ if $\pi \in W_0$, and all matrices $A^e$ lie in $Mod$.
\end{defi}


\begin{rk}\rm
Note that $sW_0$ and $sMod$ are not vector spaces. The condition
$\mathbf{1}^t  \pi \neq 0$ in the above definition means that
the sum of the coordinates of $\pi$ is not zero. Since vectors
in $sW_0$ with non-negative coordinates represent the probability distributions for the set of observations $B$, this
condition implies no restriction from a biological point of
view. Moreover, it ensures that $W_0\cap \{\sum_{\x \in
  B}\pi_{\x}=1\}$ has  dimension equal to $dim(W_0)-1$. In
particular, the simplex of stochastic vectors in
$W_0$ will form a semialgebraic set of $\langle B\rangle_{\mathbb{R}}$ of dimension equal to $dim(W_0)-1$ (as expected).
\end{rk}

\begin{rk}\rm
The subspace $Mod$ of substitution matrices is usually required to be
multiplicatively closed (as in the definition above) so that when two
evolutionary processes are concatenated, the final process is of the
same kind.
The importance of this requirement is the starting point of \cite{LMM}, where a different approach to the definition of ``evolutionary mode'' is provided.
\end{rk}

Our definition of evolutionary models includes most
of the well-known evolutionary models, namely those given in
\cite{Allman2006b} and the \emph{equivariant} models  (see\cite{Draisma,CFS3}).

\begin{ex}\label{def_equivariant}
Let $G$ be a permutation group of $B$, that is, a group whose elements are permutations of the set $B$, $G\leq \mathfrak{S}_k$.
Given $g\in G$, write $P_g$ for
the $k\times k$-permutation matrix corresponding to $g$:
$(P_g)_{i,j}=1$ if $g(j)=i$ and $0$ otherwise.  The
$G$-\emph{equivariant evolutionary model} $\mathcal{M}_G$ is defined by taking
$Mod$ equal to
\begin{eqnarray*}
M(G)=\{A\in M_{k}(\CC)\mid P_{g}AP_{g}^{-1}=A \mbox{ for all } g \in
G\},
\end{eqnarray*}
and $W_0=\{\pi\in W \mid P_{g}\pi=\pi \mbox{ for all } g\in G\}$. These subsets are vector subspaces of $M_k(\CC)$ and
$W$, respectively. Moreover, if $A_1,A_2\in M(G)$, then
\begin{eqnarray*}
 P_{g} A_1A_2 P_{g}^{-1}=(P_{g} A_1 P_{g}^{-1})(P_{g} A_2 P_{g}^{-1})=A_1A_2,
\end{eqnarray*}
and $A_1 A_2\in M(G)$. Therefore, equivariant models provide
a wide family of examples of algebraic evolutionary models in the
sense of Definition \ref{em}. For example, if $B=\{\a,\c,\g,\t\}$,
it can be seen that the algebraic versions of the Jukes-Cantor model \cite{JC69}, the Kimura models with 2 or 3 parameters \cite{Kimura1980,Kimura1981}, the strand symmetric model \cite{CS} or the general Markov model \cite{Chang96} are instances of equivariant models:
 \begin{itemize}
 \item if $G=\mathfrak{S}_{4}$, then ${\MM}_G$ is the \emph{algebraic Jukes-Cantor model} $\JC$,
 \item if $G= \langle (\a\c\g\t),(\a\g) \rangle $, then ${\MM}_G$ is the \emph{algebraic Kimura 2-parameter model} $\Kb$,
\item if $G=\langle (\a\c)(\g\t),(\a\g)(\c\t) \rangle$, then ${\MM}_G$ is the \emph{algebraic Kimura 3-parameter model} $\Ka$,
 \item if $G=\langle(\a\t)(\c\g)\rangle$, then  ${\MM}_G$ is known as the \emph{strand symmetric model} $\SSM$, and
 \item if $G= \langle e \rangle$, then ${\MM}_G$ is the \emph{general Markov model} $\texttt{GMM}$.
 \end{itemize}
\end{ex}

Given an evolutionary model $\MM$ and a phylogenetic tree $T$, we define the \emph{space of parameters} as
\[\rep_{\MM}(T)=W_0\times \left (\prod_{e\in E(T)} Mod \right ).\]
Similarly, we define the space of \emph{stochastic} parameters associated to $T$ by  \[\rep_{s\MM}(T)=sW_0\times \left (\prod_{e\in E(T)} sMod \right ).\]

Though artificial at  first glance, the use of tensors in the framework that includes the
distributions on the set of patterns in $B$ at the leaves of a
phylogenetic tree is a natural choice.
Indeed, if $p_{\x_1 \x_2 \ldots \x_n}$ denotes the joint probability of observing $\x_1$ at leaf 1, $\x_2$ at leaf 2, and so on, up to $\x_n$ at leaf $n$, then the vector $p=(p_{{b_1 \ldots b_1}}, p_{{b_1 b_1 \ldots b_2}}, \dots,p_{b_k\ldots b_k})$ provides a distribution on the set of
patterns in  $B$ at the leaves of $T$, and this can be regarded as the tensor having these coordinates in the natural basis,
\begin{eqnarray*} p=
\sum_{\x_1 \ldots \x_n \in B}p_{\x_1 \ldots \x_n} \x_1 \otimes
\ldots \otimes \x_n.
\end{eqnarray*}
This motivates the following definition.
\begin{defi}
Given a phylogenetic tree $T$ on the set of taxa $[n]$, an
$[n]$-\emph{tensor} is any element of the tensor power
\[\LL:=\otimes_{[n]}W.\]
\end{defi}

Given an algebraic evolutionary model $\MM$ and a phylogenetic tree $T$ with root $r$, every Markov process on $T$ (specified by a collection of parameters $\pi$ and $\mathbf{A}=(A_e)_{e\in E(T)}$) gives rise to a tensor in $\LL$ in the following way: %
we consider a \emph{parametrization}
\begin{eqnarray}\label{param}
\Psi_T^{\MM}: \rep_{\MM}(T)\lra \LL
\end{eqnarray}
defined by
\begin{eqnarray*}
\Psi_T^{\MM}\left(\pi, \mathbf{A}\right)=\sum_{\x_i \in B}p_{\x_1...\x_n}\x_1 \otimes \dots \otimes \x_n,
\end{eqnarray*}
where
\begin{eqnarray}\label{form_param}
p_{\x_1\dots \x_n}=\sum_{\x_v \in B, v \in Int(T)}
\pi_{\x_r}\prod_{e\in E(T)}
A^{e}_{\x_{pa(e)}, \x_{ch(e)}} \, ,
\end{eqnarray}
$\x_v$ denotes the state at the vertex $v$,
$pa(e)$ (resp. $ch(e)$) is the parent (resp. child) node of $e$, and $\pi_{\x}, \x \in
B$, are the coordinates of $\pi$.
When restricted to the stochastic matrices and distributions in $W_0$,
this parametrization corresponds to the hidden Markov process on the tree $T$
(the leaves correspond to the observed random variables and the interior
nodes to the hidden variables).

The parametrization  (\ref{param}) restricts to another polynomial map  $\phi_T^{\MM}:  \rep_{s\MM}(T) \longrightarrow H
$, where $H \subset \LL$  is the hyperplane defined by $H=\left\{ p \in \LL \mid  \sum_{\x_1, \dots, \x_n\in B} p_{\x_1 \dots \x_n}
=1 \right\}.$
Because we work in the algebraic setting, the use of the word ``stochastic'' in this paper is more general than usual, as we only request entries summing to one.

From now on, we will refer to this restriction  as the \emph{stochastic
parametrization}  $\phi_T^{\MM}$. It is important to note that
when we consider the distributions in $sW_0$ and the Markov matrices
in $sMod$, its image by  $\phi_T^{\MM}$ lies in the standard
simplex in $\LL$ (and thus in $H$). This in turn implies that the
whole image ${\im} \, \phi_T^{\MM}$ is contained in $H$.

We proceed to define the algebraic varieties associated to the
parametrization maps defined above. Roughly speaking, algebraic
varieties are sets of solutions to systems of polynomial equations
(e.g. \cite{Harris}).

\begin{defi}
The \emph{stochastic phylogenetic variety $V_T^{\MM}$
associated to a phylogenetic tree }$T$  is the smallest algebraic variety containing ${\im} \, \phi_T^{\MM}=\left\{ \phi_T^{\MM} (\pi_r,\mathbf{A}) : (\pi_r,\mathbf{A})\in  \rep_{s\MM}(T) \right\}$ (in particular, $V_T^{\MM}\subset H$).

Similarly, the \emph{phylogenetic variety $CV_T^{\MM}$ associated
to }$T$ is the smallest algebraic variety in $\LL$ that contains ${\im}
\, \Psi_{T}^{\MM}=\left\{ \Psi_T^{\MM} (\pi_r,\mathbf{A}) :
  (\pi_r,\mathbf{A})\in  \rep_{\MM}(T) \right\}$.
\end{defi}
Below we explain the reason for the notation of $CV_T^{\MM}$, which was adopted from \cite{Allman2004b}.

The reader may note that the position of the root $r$ of $T$ played a role in the above
parameterizations.
It can be shown, however, that under
 certain mild assumptions, ${\im} \, \Psi_T^{\MM}$ and ${\im} \, \phi_T^{\MM}$ are independent of the root position in the following sense: if two phylogenetic trees have the same topology as unrooted trees, then the smallest algebraic varieties containing the corresponding image sets are the same.  
For example, any model $\MM=(W_0,Mod)$  satisfying (i) $\widetilde{\pi}^t:=\pi^t A$ belongs to $W_0$ for all $\pi \in
W_0$ and all $A \in Mod,$ and (ii) $D_{\widetilde{\pi}}^{-1} A^t D_{\pi} \in Mod$ whenever $D_{\widetilde{\pi}}^{-1}$ exists (here $D_{\omega}$ denotes the diagonal matrix with the entries of $\omega$  on the diagonal and zeros elsewhere) has this property (in this case, we say the model is \emph{root-independent}).
It is not difficult to check that the equivariant
models satisfy these two properties (e.g. adapting the proof
of \cite{Allman2003} or \cite{Steel98}). For technical reasons, from
now on we consider only the evolutionary models satisfying (i) and (ii).
Indeed, in this case the notation $CV_T^{\MM}$ refers to the fact that
the phylogenetic variety is just the cone over the stochastic
phylogenetic variety (see Figure \ref{fig_1} and the remark below).

\begin{figure}
\begin{center}
\includegraphics[scale=0.2]{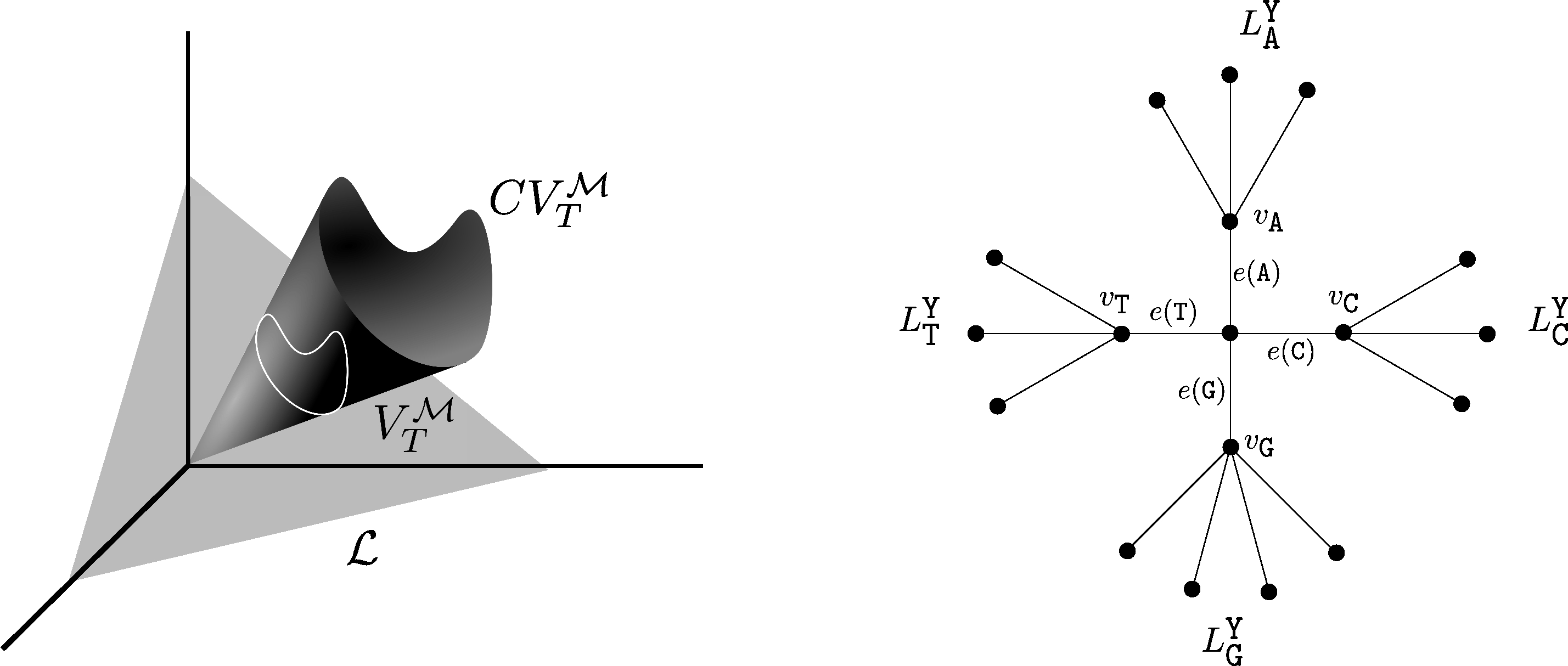}
\end{center}
 \caption{\label{fig_1} On the left, the varieties $V_T^{\MM}$ and $CV_T^{\MM}$ are shown; on the right, the phylogenetic tree described in the proof of Proposition \ref{triv_st} is represented.}
\end{figure}

\begin{rk} \label{homog_parametr}
\rm
Let $\MM$ be an evolutionary model satisfying (i) and (ii) above. For
$p\in \LL$, $p= \sum p_{\x_1 \ldots \x_n} \x_1\otimes \ldots \otimes
\x_n$, define $\lambda(p):=\sum_{\x_i\in B}p_{\x_1 \ldots
  \x_n}$. Then $$CV_T^{\MM}=\left\{p \in \LL | \, p=\lambda(p) q \, , q \in V_T^{\MM}
\right\}$$
and  $ V_ T^{\MM}=CV_T^{\MM}\cap H$.
This is well known for the general Markov model \cite{Allman2004b} and
can be easily generalized to any model satisfying (i) and (ii).
\end{rk}

\subsubsection*{The space of phylogenetic mixtures}

In phylogenetics, the hypothesis that the sites of an alignment
are  independent and identically distributed is often used. When the assumption ``identically
distributed'' is replaced by ``distributed according to the
same evolutionary model'', one obtains a phylogenetic mixture.
Below, we introduce phylogenetic mixtures from the algebraic point of view (see also \cite{MatsenMosselSteel}).

\begin{defi}
Fix a set of taxa $[n]$ and an algebraic evolutionary model $\MM$.
A \emph{phylogenetic mixture (on $m$-classes)} or
\emph{$m$-mixture} is any vector $p \in \LL=\otimes_{[n]}W$ of
the form
\[p=\sum_{i=1}^m \alpha_i p^i,  \]
where $\alpha_i \in \mathbb{C}$ and $p^i\in {\im} (\Psi_{T_i}^{\MM})$ for some tree topologies $T_i$ on the set of taxa $[n]$. As $ \Psi_{T_i}^{\MM}$ is a homogeneous
map, phylogenetic mixtures are  represented by vectors of the form
$\sum_{i=1}^m \check{p}^i$, where $\check{p}^i\in
{\im}(\Psi_{T_i}^{\MM}) $.
We call $\DD_{\MM} \subset \LL$ the \emph{space of all phylogenetic
mixtures} (on any number of classes) under the algebraic
evolutionary model $\MM$.%

As mentioned in the introduction, the tree topologies
  contained in the mixture can be the same or different. An example of a
  phylogenetic mixture is the data modeled by the discrete
  Gamma-rates models (see e.g.
\cite{Semple2003}).

Restricting matrix rows to sum to one requires
restricting the phylogenetic mixtures to the points of the form
\[q=\sum_{i=1}^m \alpha_i q^i \quad \textrm{ where } \quad q^i\in {\im} (\phi_{T_i}^{\MM}) \, , \textrm{ and } \, \sum_i \alpha_i=1.\]
We call $\DD_{s\MM}$ the space of \emph{stochastic phylogenetic
mixtures}.
\end{defi}

\begin{rk}\rm
The phylogenetic variety of a trivalent tree topology contains all
phylogenetic varieties of the non-trivalent tree topologies obtained
by contracting any of its interior edges. Indeed, the latter are a
particular case of the former when the matrices associated to the contracted edges are equal to the
identity matrix. It follows that the space of phylogenetic
mixtures on the trivalent tree topologies coincides with the space of phylogenetic mixtures on all possible topologies.
\end{rk}

The following result was proven by Matsen, Mossel  and Steel in
\cite{MatsenMosselSteel} for the two state random cluster model but, as proved below, it can be easily generalized to any  evolutionary model.

\begin{lema}\label{DsmDm}
Given a set of taxa $[n]$ and an algebraic evolutionary model
$\MM$, the set of all phylogenetic mixtures $\DD_{\MM}$ is a
vector subspace of $\LL$. Similarly, $\DD_{s\MM}$ is a linear variety and it equals $ \DD_{\MM} \cap H$.
\end{lema}

\begin{proof}
$\DD_{\MM}$ is a $\mathbb{C}$-vector space and $\DD_{s\MM}$ is a linear variety by their definition.
It follows that $\DD_{\MM}$ is an algebraic variety that contains ${\im} \, \Psi_{T}^{\MM}$ for any phylogenetic tree $T$ on the set of taxa $[n]$. Therefore, it also contains $CV_{T}^{\MM}$, and
$\DD_{\MM}$ equals the set of points of the form $p=\sum p_i$, where $p_i\in CV_{T_i}^{\MM}$. Similarly,  $\DD_{s\MM}$ is an algebraic
variety that contains ${\im} \, \phi_{T}^{\MM}$, so it also contains $V_{T}^{\MM}$ for any phylogenetic tree $T$. It follows that
$\DD_{s\MM}$ is formed by points of type $q=\sum \alpha_i q_i$, where $q_i\in V_{T_i}^{\MM}$ and ${\sum_i \alpha_i=1.}$

Now we check that $\DD_{s\MM}=\DD_{\MM}\cap H$. Let $q\in   \DD_{s\MM}$, so that
$q=\sum_{i=1}^m \alpha_i q^i$ for some $m$,  $q^i \in V_{T_i}^{\MM}$, and
$\sum \alpha_i=1$. Clearly, $q\in \DD_{\MM}$. Moreover,
the sum of coordinates of $q$, $\lambda(q)$, satisfies
$\lambda(q)=\sum_i \alpha_i \lambda(q^i)=\sum_i \alpha_i =1$.
Thus, $q\in H$.
Conversely, let $p= \sum_{i=1}^m p^i $ with $p^i\in
CV_{T_i}^{\MM}$ for certain tree topologies $T_i$, and assume that
$\lambda(p)=1$. Apply Remark  \ref{homog_parametr} to
each $p^i$ to get $p^i=\lambda(p^i) q_i$ for some $q_i \in
V_{T_i}^{\MM}$. Then
\begin{eqnarray*}
 p=\sum_i p^i=\sum_i \lambda(p^i)q_i
\end{eqnarray*}
and $1=\lambda(p)=\sum_i \lambda(p^i)\lambda(q_i)=\sum_i
\lambda(p^i)$ since each $q_i$ lies on $H$. This proves that $p\in \DD_{s\MM}$.
\end{proof}

\begin{rk}\rm
 In the proof of the above lemma, we have seen that  $\DD_{\MM}$
 and $\DD_{s\MM}$ can be alternatively described as the spaces
 of mixtures obtained from the respective varieties $CV^{\MM}_{T}$ and
 $V^{\MM}_{T}$ (i.e., not only from the images of the parametrization
 maps).
\end{rk}

\subsection*{The space of phylogenetic mixtures for equivariant
evolutionary models}\label{secteq}

This section  provides a precise description of the
space $\DD_{\MM}$ for the equivariant models $\MM$ listed in Example
\ref{def_equivariant} ($\JC$, $\Kb$, $\Ka$, $\SSM$, and
$\texttt{GMM}$). First, we recall some
definitions and facts of group theory and linear representation theory.  From now on, $B=\{\a,\c,\g,\t\}$,  $k=4$, $W=\langle B \rangle_{\mathbb{C}}$, $n$ is fixed and $\LL=\otimes_{[n]} W$.

\subsubsection*{Background on representation theory}

We introduce some tools in group representation theory needed in the sequel. We refer the reader to \cite{Serre} as a classical reference for these concepts. Although some of the following results are valid for any permutation group, for simplicity in the exposition we restrict to permutations of four elements (as our applications deal only with the case $B=\{\a,\c,\g,\t\}$).

Let $G\leq \mathfrak{S}_4$ be a permutation group. The trivial element in $\mathfrak{S}_4$ will be denoted as $e.$
We write $\rho_G$ for the restriction to $G$
of the \emph{defining} representation $\rho:\mathfrak{S}_4\ra GL(W)$
given by the permutations of the basis $B$ of $W$. %
This representation induces a $G$-module structure on $W$ by
setting $g\cdot \x:=\rho(g)(\x)\in W$. In fact, $\rho$ induces a $G$-module structure on any tensor power $\otimes^s W$ by setting
\begin{eqnarray}\label{mult_action}
g\cdot \left (\x_{1}\otimes \ldots \otimes \x_{s}\right ):=g\cdot
\x_{1} \otimes \ldots \otimes g\cdot \x_{s},
\end{eqnarray}
and extending by linearity.
From now on, the space $\LL=\otimes^n W$ will be implicitly considered
as a $G$-module with this action.
We call $\chi$ the character associated to the representation $\rho_G:
G \rightarrow GL(W)$, i.e. $\chi(g)$ is the trace of the
corresponding permutation matrix or, in other words, $\chi(g)$ equals
the number of fixed elements in $B$ by the permutation $g\in G$. Then
the character associated to the induced representation $G \rightarrow GL(\otimes^n W)$ is $\chi^n$, the $n$-th power of $\chi$.

We write $N_1, \ldots, N_t$  for the irreducible representations
of $G$ and $\omega_1, \ldots, \omega_t$ for the corresponding
irreducible characters, where $N_1$ and $\omega_1$ will denote the trivial
representation and trivial character, respectively.
Maschke's Theorem applied to the action of $G$ described in (\ref{mult_action}) states that there is a
decomposition of $\otimes^s W$  into its isotypic components:
\begin{eqnarray}\label{Marsche}
\otimes^s W = \oplus_{i=1}^t (\otimes^s W)[\omega_i],
\end{eqnarray}
where each $(\otimes^s W)[\omega_i]$ is isomorphic to  a number of
copies of the irreducible representation $N_i$ associated to
$\omega_i$, $(\otimes^s W)[\omega_i]\cong N_i\otimes
\mathbb{C}^{m_i(s)}$, for some non-negative integer $m_i(s)$ called
the \emph{multiplicity of $\otimes^{s} W$ relative to
$\omega_{i}$}.
The isotypic component of $\LL$ associated to the trivial
representation will be denoted by $\LL^G$ and it is composed of the
$[n]$-tensors invariant under the action of $G$ defined in
(\ref{mult_action}). 
If $\MM$ is the equivariant evolutionary model associated to $G$, $\LL^{G}$ will also be denoted as $\LL^{\MM}$. It is easy to prove that $CV^{\MM}_T\subset \LL^G$ (see Lemma 4.3 of \cite{Draisma}).

We  recall that the set $\Omega_G=\{\omega_i\}_{i=1,\ldots,t}$ of irreducible characters of $G$ forms an orthonormal
basis of the space of characters relative to the inner product
defined by
\begin{eqnarray}\label{in_prod}
\langle f,h\rangle :=\frac{1}{|G|}\sum_{g\in G}f(g)\overline{h(g)}.
\end{eqnarray}

We introduce the following notion.

\begin{defi}\label{notat_B}
An $n$-\emph{word over }$B$ is an ordered sequence $\X=\x_1 \x_2 \ldots \x_n$, where every letter is taken from the alphabet $B$. The set of $n$-words is equivalent to the cartesian power $B^n$ and will be denoted by $\B$.
\end{defi}
Words will be denoted in typewritter uppercase font (like $\X$) and their letters in lowercase (like $\x$).
Sometimes it will be convenient to identify the $[n]$-tensors of the form $\x_1\otimes \ldots \otimes \x_n$ with the $n$-words $\X=\x_1 \ldots \x_n$. Consequently, we will identify $\B$ with the natural basis of $\LL$. 
Given $\X\in \B$, we will denote by $\{\X\}_{G}=\{g \X\mid  g\in G\}$
the $G$-\emph{orbit of} $\X$.  We associate a $G$-invariant tensor, $\tau\{\X\}_G$, to each orbit $\{\X\}_{G}$: $\tau\{\X\}_G :=\sum_{g\in G}g \X.$
It is straightforward to see that every $G$-invariant tensor
can be written as a linear combination of the tensors
$\tau\{\X\}_G$, $\X \in \B$. On the other hand, the set of different
$\tau\{\X\}_G$'s is linearly independent, since the corresponding
$G$-orbits $\{\X\}_G$ have non-overlapping composition of the  elements of $\B$.

\subsubsection*{Mixtures for equivariant models}

For each $\x\in B$, we write $S_G(\x)$ for the \emph{stabiliser of }$\x$ under the action of $G$, that is, $S_G(\x)={\{g\in G: g \cdot \x=\x\}}$. %

\begin{prop}\label{triv_st}
 Let $G$  be a subgroup of $\mathfrak{S}_4$ such that  $S_G(\x_0)=\{e\}$  for some $\x_0\in B$. Then every tensor of type  $\tau\{\X\}_G$, $\X \in \B$, lies in the image of $\Psi_T^{{\MM}_G}$  for some tree topology $T$. In particular, $\LL^{G}\subset \mathcal{D}_{\MM_G}$.
\end{prop}

\begin{proof}
For any $G$-orbit $\{\texttt{Y}\}_G$, $\texttt{Y} \in \B$, write  $\tau\{\texttt{Y}\}_G=y_{1} \otimes \ldots \otimes  \y_{n}+\sum_{g\neq e} g\cdot
\y_{1} \otimes \ldots \otimes g\cdot \y_{n}$. We will explicitly
associate a tree topology and parameters $(\pi,\mathbf{A})$ to it so that the tensor $\tau\{\texttt{T}\}_G$ is equal to $\Psi_T^{{\MM}_G}$.
To this aim, we denote by $B(\mathtt{Y})$ the set of letters appearing in $\mathtt{Y}$. Then for every $\z\in B(\mathtt{Y})$, consider the set $L^{\mathtt{Y}}_{\z}=\{i\in [n]: \y_i=\z\}$, so that $\cup_{\z\in B(\mathtt{Y})}L^{\mathtt{Y}}_{\z}=[n]$.

We construct a tree $T$ on the set of taxa $[n]$ in the following way. We join each taxa in $L^{\mathtt{Y}}_{\z}$ to a common node $v_{\z}$ by an edge.  
Then each vertex $v_{\z}$ is joined to the root of the tree (we call it $r$) by an edge that we denote as $e(\z)$ (see figure \ref{fig_1}).
Now, in the edges joining any $v_{\z}$ with some leaf in $L^{\mathtt{Y}}_{\z}$, we consider the identity matrix, while the matrix in $e(\z)$ is defined by taking
\begin{eqnarray*}
 A^{e(\z)}_{i,j}=\left \{ \begin{array}{ll} 1 & \mbox{ if }(i,j)=(h\cdot \x_0, h\cdot \z) \mbox{ for some }h\in G , \\ 0 & \mbox{otherwise.}  \end{array} \right.
\end{eqnarray*}
Finally, if $c$ is the cardinality of $\{x_0\}_G$, define the distribution at the root $\pi=(\pi_{\a},\pi_{\c},\pi_{\g},\pi_{\t})$ by
\begin{eqnarray*}
 \pi_{\z}=\left \{ \begin{array}{cl} \frac{1}{c} & \mbox{ if }\z\in \{\x_0\}_G, \\ 0 & \mbox{otherwise.}  \end{array} \right.
\end{eqnarray*}
It is straightforward to check that these matrices and the vector $\pi$ are $G$-equivariant, so $(\pi,\mathbf{A})\in \rep_{{\MM}_G}(T)$. Now, from (\ref{form_param}) and the definition of $\pi$, we can write
\begin{eqnarray*}
 p_{\x_1\dots \x_n}=\sum_{\substack{g\in G\\
\{\x_{\z}\}_{\z \in B(\mathtt{Y})}\subset B
} }P_{\x_1\dots \x_n}(g,\{\x_{\z}\}_{\z \in B(\mathtt{Y})})
\end{eqnarray*}
where
\begin{eqnarray*}
P_{\x_1\dots \x_n}(g,\{\x_{\z}\}_{\z \in B(\mathtt{Y})})=\pi_{g \cdot \x_0}
\prod_{{\z}\in B(Y)}
\left ( A^{e(\z)}_{g \cdot \x_0,\x_{\z}} \prod_{j\in L^{\mathtt{Y}}_{\x_{\z}}} \delta_{\x_{\z},\x_j} \right )
\end{eqnarray*}
(here $\delta_{a,b}$ stands for the Kronecker delta, i.e.  $\delta_{a,a}=1$, $\delta_{a,b}=0$ if $a\neq b$).
Moreover, from the definition of the matrix $A^{e(\z)}$, we have
\begin{eqnarray*}
 A^{e(\z)}_{g \cdot \x_0,\x_{\z}}=\left \{ \begin{array}{ll} 1 & \mbox{ if }(g  \cdot\x_0,\x_{\z})=(h\cdot \x_0, h\cdot \z) \mbox{ for some }h\in G , \\ 0 & \mbox{otherwise.}  \end{array} \right.
\end{eqnarray*}
The hypothesis $S_G(x_0)=\{e\}$ ensures that $(g \cdot \x_0,\x_{\z})=(h\cdot \x_0, h\cdot \z)$ if and only  if $g=h$.
From this, it becomes clear that
$P_{\x_1\dots \x_n}(g,\{\x_{\z}\}_{\z\in B(Y)})=0$ unless
\begin{enumerate}
 \item $\x_{\z}= g \cdot \z$, for $\z\in B$, and
 \item for each $i\in L_{\z}^{\mathtt{Y}}$, $\x_i$ is equal to $\x_{\z}=g\cdot \z$,
\end{enumerate}
in which case $P_{\x_1\dots \x_n}(g,\{\x_{\z}\}_{\z\in B(Y)})=\pi_{g\cdot x_0}=\frac{1}{c}$.
It follows that
\begin{eqnarray*}
p_{\x_1\dots \x_n}=
\left \{ \begin{array}{cl} 1 & \mbox{ if } \x_1\dots \x_n\in \{\mathtt{Y}\}_G , \\ 0 & \mbox{otherwise,}  \end{array} \right.
\end{eqnarray*}
and $\Psi_T^{\MM}(\pi,\mathbf{A})=\tau\{\mathtt{Y}\}_G$. Moreover, as the set of $\tau \{\mathtt{Y}\}_{G}$, for $\mathtt{Y}\in \B$, generates the vector space $\LL^G$, the second claim follows.
\end{proof}

\begin{rk}\rm
 The above result is not true if the hypothesis $S_G(\x_0)=\langle e \rangle$ is removed. For example, if $G=\langle (\a\c\g\t),(\a\g) \rangle$ (so that $\MM=\Kb$), then $S_G(\a)=S_G(\g)=\{e,(\c\t)\}$ and $S_G(\c)=S_G(\t)=\{e,(\a\g)\}$. In that case, it can be shown that the $G$-orbit $\{\a\c\g\t\}_G$ is not in ${\im} \Psi_T^{\Kb}$ for any tree topology $T$ with 4 leaves.
\end{rk}

Since the above condition on the group holds for  $G=\mathfrak{S}_4$,
$G=\langle (\a\t)(\c\g) \rangle$, and $G=\langle
(\a\g),(\a\c\g\t)\rangle$, we deduce the following claim.

\begin{cor}\label{partial_models}
If $G$ corresponds to any of the equivariant models \texttt{K81}, \texttt{SSM} or
\texttt{GMM}, we have $\LL^{G}\subset \mathcal{D}_{\MM_G}$.
\end{cor}

In phylogenetics, an \emph{invariant of a phylogenetic tree $T$} is an
equation satisfied by the expected distributions of patterns at the
leaves of $T$, irrespectively of the continuous
parameters of the model $\MM.$ In the algebraic geometry setting, these
are the equations satisfied by all   $p\in CV^{\MM}_T.$
Invariants were introduced by Lake (see \cite{Lake1987}) and Cavender
and Felsenstein (see \cite{Cavender87}).
A \emph{phylogenetic
  invariant of $T$} is an invariant of $T$, which is not
an invariant of all other phylogenetic trees (under the same
model $\MM$). Equivalently, $f$ is a phylogenetic invariant of
$CV^{\MM}_T$ if it is an invariant of $CV^{\MM}_T$ and there exists a
tree topology $T'$ such that $f$ is not an invariant of
$CV^{\MM}_{T'}$. In principle, phylogenetic invariants can be used for tree topology reconstruction purposes.

\begin{rk}\label{linear_inv} \rm
\begin{enumerate}
 \item[(a)]  It can be seen that the condition of trivial stabiliser
   for some element of $B$ given  in Proposition \ref{triv_st} guarantees that all the irreducible representations of $G$ will be present in the decomposition of $W$ into its isotypic components. Then, by using the results of \cite{CFS3}, it follows that the corresponding equivariant model will have no linear phylogenetic invariants. This fact
was already known for the models in the above corollary: see
\cite{Allman2004} for the  \texttt{GMM},  \cite{CS} for the
\texttt{SSM} and \cite{Sturmfels2005} for the \texttt{K81}. Here we
provided an alternative proof based on elementary tools of group theory.

\item[(b)] The models \texttt{JC69} and \texttt{K80} are known to have
  linear phylogenetic invariants, but these are the only linear
  invariants which do not define hyperplanes containing $\LL^G$, as
  can be deduced from \cite{FuLi,Sturmfels2005}. In fact, for these
  two models, the claim of the corollary is still true as stated
  in the following theorem. Nevertheless, we have not been able to provide a unified proof of this fact because of the different properties of the corresponding groups.
There is no description of the space of linear invariants for other
equivariant models not listed in Example \ref{def_equivariant},  so we cannot claim that the result below still
holds.
\end{enumerate}
\end{rk}

\begin{thm}\label{caractLG}
If $\MM_G$ is one of the equivariant evolutionary models
\texttt{JC69}, \texttt{K80}, \texttt{K81}, \texttt{SSM}, or
\texttt{GMM}, then the space of phylogenetic mixtures $\DD_{\MM_G}$ coincides with $\LL^{G}$,
and $\DD_{s\MM_G}$ equals $\LL^{G} \cap H$.
\end{thm}

This theorem allows to identify the set of all phylogenetic
mixtures $\DD_{\MM_G}$ with $\LL^{G}$, which is a vector
subspace of $\LL$ whose linear equations are easy to describe.
In other words, $\LL^{G}$ is the smallest linear space containing the  data coming from any mixture of trees
evolving under the model $\MM_G$. One can therefore use
$\LL^{G}$ to select the most suitable model for the given data.
This has been studied in  \cite{KDGC}.

\vspace{3mm}
\begin{proofof}{Theorem \ref{caractLG}}
For equivariant
models we have that $CV^{\MM_G}_T \subset \LL^{G}$ for
any tree $T$. Hence, by Lemma \ref{DsmDm} and the definition of $\DD_{\MM_G}$,
$\DD_{\MM_G}$ is a vector subspace of $\LL^{G}$.

From Corollary \ref{partial_models}, we infer the equality $\LL^{G}= \DD_{\MM_G}$  for the models \texttt{K81}, \texttt{SSM} and \texttt{GMM}. For the other two models, \texttt{JC69} and \texttt{K80}, it remains to prove that there does not exist any hyperplane $\Pi$ containing
$\DD_{\MM_G}$ and not containing $\LL^{G}$.
If such a
hyperplane existed, then it would contain all the
points of $CV_{T}^{\MM_G}$ for any tree topology $T$.
It suffices to prove that for these models there are no homogeneous linear polynomials vanishing on all tree topologies,
except for the linear equations vanishing on $\LL^{G}$. This has been seen in Remark \ref{linear_inv}(b).

The equality $\DD_{s\MM_G}=\LL^{G} \cap H$ follows
immediately from Lemma \ref{DsmDm}  and the first assertion in the
statement of this theorem.
\end{proofof}

\begin{rk}\rm
We are indebted to one of the referees of this paper for pointing out that the preceeding result, as well as the second part of Proposition \ref{triv_st}, can also be inferred from Proposition 4.9 of \cite{Draisma}: under the assumption that the stabiliser of some state is trivial, Draisma and Kuttler show that the star tree is the smallest algebraic variety containing the tensors $\tau \{\X\}_G$, for \emph{pure} tensors $\X$  (that is, tensors of rank 1). It follows that the set of mixtures on the star tree equals the space $\LL^G$.
\end{rk}

\begin{rk}
\rm It is not difficult to check that for $\MM=$ \texttt{K81}, \texttt{SSM} or \texttt{GMM}, $\DD_{\MM}$ coincides with the space of mixtures on the star tree (see also \cite{MatsenMosselSteel}, where the same result is proven for a 2-state model). On the contrary, this is not true for \texttt{JC69} and \texttt{K80} models because in this case the star tree lies in a smaller linear space as a consequence of the existence of phylogenetic linear equations (see Remark \ref{linear_inv}(b)).
\end{rk}

\subsubsection*{Equations for the space $\LL^{G}$}
Our goal here is to compute the dimension of $\LL^{G}$ for the groups associated to the equivariant models listed in Definition
\ref{def_equivariant}, and to list  a set of independent linear equations defining this space.

\begin{prop}\label{dim_models}
Using the notations above,
\begin{itemize}
\item[(i)] $\dim \LL^{\SSM}=2^{2n-1}$,
\item[(ii)] $\dim \LL^{\Ka} =  4^{n-1}$,
\item[(iii)] $\dim \LL^{\Kb}= 2^{2n-3}+2^{n-2}$, and
\item[(iv)] $\dim \LL^{\JC}=\frac{2^{2n-3}+1}{3}+2^{n-2}$.
\end{itemize}
\end{prop}

\begin{proof}
Let $\MM$ be any equivariant model.
By definition, we know that $\LL^G$ is the isotypic component of $\otimes^n W$ associated to the trivial representation $(\otimes^n W)[\omega_1]$.
Since the dimension of the trivial representation is one, it follows that the dimension of $\LL^{\MM}$  is precisely the multiplicity
$m_1(n)$, i.e. the number of times the
trivial representation appears in the decomposition of $\otimes ^n W$ into isotypic
components. This multiplicity $m_1(n)$ equals (see (\ref{in_prod}))
\begin{eqnarray*}
\langle \chi^n,\omega_1\rangle =\frac{1}{|G|}\sum_{g\in
G}\chi^n(g)\omega_1(g).
\end{eqnarray*}
The proof ends by grouping the elements of $G$ in the conjugacy
classes of $G$ for $\SSM$, $\Ka$, $\Kb$, or $\JC$. Recall that the
conjugacy classes of a group $G$ are the disjoint sets of the form $C(g) = \{h^{-1}gh  : h \in G\}$.  If $C_1,\ldots,C_s$ are the conjugacy classes
for $G$, write $\cc(G)=(|C_1|,\ldots,|C_s|)$ for the $s$-tuple of
their cardinalities, so that $\sum_{i=1}^s |C_i|=|G|$.
Recall that
$\chi^n(g_1)=\chi^n(g_2)$  whenever $g_{1}$ and $g_{2}$ lie in
the same conjugacy class, so we can represent $\chi^n$ by an $s$-tuple
$\chi^n_{\cc(G)}=(t_1,\ldots,t_s)$, where $t_i=\chi^n(g)$ for any $g\in C_i$. Thus, we have $m_1(n)=\frac{1}{|G|}\sum_{i=1
}^s \chi^n(g_i)| C_i| $, where $g_i$ is any element in the conjugacy class $C_i$. The result for $\MM=\SSM$, $\Ka$, $\Kb$, or $\JC$ follows by applying the following table.

\small
\begin{table}[here]
\begin{tabular}{lllll} \label{examples_groups}
$G\leq \mathfrak{S}_{4}$ & $\MM$ & representatives of conj. classes & $\cc(G)$ &  $\chi^n_{\cc(G)}$ \\ \hline

$\langle (\a\t)(\c\g)\rangle $ & $\SSM$ & $\{e,(\a\t)(\c\g)\}$& $(1,1)$ & $(4^n, 0)$  \\

   $\langle (\a\c)(\g\t),(\a\g)(\c\t) \rangle $ & $\Ka$ & $\{e,(\a\t)(\c\g),(\a\c)(\g\t),(\a\g)(\c\t)\}$ &  $(1, 1, 1, 1)$ & $(4^n, 0, 0, 0)$  \\

 $\langle (\a\c\g\t),(\a\g) \rangle $ & $\Kb$ & $\{e,(\a\c)(\g\t),(\a\g)(\c\t),(\a\c\g\t),(\a\g)\}$ & $(1, 2, 1, 2, 2)$ & $(4^n, 0, 0, 0, 2^n)$  \\

  $\mathfrak{S}_{4} $ & $\JC$  & $\{e, (\a\c)(\g\t),(\a\c\g\t), (\a\g),(\a\c\g)\}$ & $(1, 3, 6, 6, 8)$ & $(4^n,0,0,2^n,1)$
\end{tabular}
\end{table}
\end{proof}

Our next goal is to provide a set of independent linear equations
for $\LL^{G}$. Before stating the main result, let us introduce
some useful notation.

\begin{notat}\rm
We consider the following subsets of $\B=B^n$:
\begin{eqnarray*}
\B_0 & = & \{\a\ldots\a,\c\ldots\c,\g\ldots\g,\t\ldots\t \},  \\
\B_{\a\c\mid \g\t} & = &\{\a,\c\}^{n}\cup \{\g,\t\}^{n}, \\
\B_{\a\g\mid \c\t} & = &\{\a,\g\}^{n}\cup \{\c,\t\}^{n}, \\
\B_{\a\t\mid \c\g} & = &\{\a,\t\}^{n}\cup \{\c,\g\}^{n}, \mbox{ and} \\
\B_2 & = &\B_{\a\c\mid \g\t}\cup \B_{\a\g\mid \c\t} \cup \B_{\a\t\mid \c\g}.
\end{eqnarray*}
The set $\B_0$ is composed of all $n$-words with only one letter and it is contained in $\B_{\a\c\mid \g\t}$, $\B_{\a\g\mid \c\t}$, and $\B_{\a\t\mid \c\g}$. Similarly, $\B_2$ is composed of all $n$-words with two letters at
most.
It is straightforward to check that
$|\B_{\a\c\mid \g\t}|=|\B_{\a\g\mid \c\t}|=|\B_{\a\t\mid \c\g}|=2^{n+1}$ and $|\B_2|=3\cdot 2^{n+1}-8$.

We  adopt multiplicative notation for $n$-words, for instance, we write $\c^l$ for the word  
$\underbracket[1pt]{\c\ldots \c}_l$, and $(\a^l)(\g^m)\x_{l+m+1}\ldots
\x_n$ for
$\underbracket{\a\ldots \a}_l\underbracket[1pt]{\g\ldots \g}_m \x_{l+m+1}\ldots \x_n$, where $\x_{l+m+1},\ldots, \x_n$ represent any letters.
\end{notat}

The main result of this section is the following:
\begin{thm}\label{teo_indep_equations}
A set of linearly independent equations $\mathbb{E}^{\MM}$ defining
$\LL^{\MM}$ for $\MM=$ $\JC$, $\Kb$, $\Ka$, or $\SSM$ is given by
\begin{itemize}

\item[$\mathbb{E}^{\SSM}$] : equations $p_{\X}=p_{(\a\t)(\c\g)\X}$ for all $\X \in \B$ with  $\x_1\in \{\a,\c\}$;

\item[ $\mathbb{E}^{\Ka}$] : the equations in $\mathbb{E}^{\SSM}$, and the equations
$\{p_{\X}=p_{(\a\c)(\g\t)\X}\}$
for all $\X\in \B$ with $\x_1=\a$;

\item[$\mathbb{E}^{\Kb}$] : the equations in $\mathbb{E}^{\Ka}$, plus the equations
$\{p_{\X}=p_{(\a\g)\X}\}$
for all $\X\in \B\setminus \B_{\a\c\mid \g\t}$
having $\x_1=\a$ and satisfying the following condition: if
$\t$ appears in $\X$, then there is
a $\c$ in a preceding position;

\item[$\mathbb{E}^{\JC}$] : the equations in $\mathbb{E}^{\Kb}$, together with the equations
$\{ p_{\X}=p_{(\a\t)\X} \}$
for all  $\X\in \B_{\a\c\mid \g\t}\setminus \B_0$ of the form $(\a^{l})(\c^{m})\x_{l+m+1}\ldots\x_n$; plus the equations
$\{ p_{\X}=p_{(\a\c)\X}\}$ and $\{p_{\X}=p_{(\a\t)\X}\}$
for all $\X\in \B\setminus \B_2$ of the form
$(\a^{l})(\c^{m})\x_{l+m+1}\ldots\x_n$ and satisfying the condition: if $\t$ appears in
$\X$, then there is a $\g$ in a preceding position.
\end{itemize}

The number of equations added in each case is
$2^{2n-1}$ for $\SSM$,  $2^{2n-2}$ for $\Ka$, $2^{2n-3}-2^{n-2}$ for $\Kb$, and $2^{n-1}-1 + 2(\frac{2^{2n-3}+1}{3}-2^{n-2})$ for $\JC$.
\end{thm}

Before proving this theorem, we explain how these sets of equations were obtained. Notice that a system of linear equations of $\LL^G$  is given by
\begin{eqnarray*}
 \left \{ p_{g \X}=p_{\X}\mid g\in G ,\; \X\in \B \right \}.
\end{eqnarray*}
The role played by the $G$-orbits on $\B$ becomes apparent. Indeed, the idea
is  to relate the equations to the orbits of a subgroup of $G$. To this aim, let $H$ be a subgroup of $G$ and write $H\setminus G  = \{Hg : g\in G\}$ for the set of right cosets of $H$ in $G$. 
%
%
We consider a transversal of $H\setminus G$, i.e. a collection $\{g_1,\ldots, g_{[G:H]}\}$ such that
$
G=\bigsqcup_{i=1}^{[G: H]} Hg_i.
$
Then the orbit of any $\X\in \B$ can be decomposed as
\begin{eqnarray}\label{orb}
 \{\X\}_{G}=\bigcup_{i=1,\ldots, [G:H]} \{g_i\X\}_{H}.
\end{eqnarray}
This decomposition establishes the connection between the $G$-orbits and the $H$-orbits. 
%
In order to  obtain a system of equations for $\LL^G$, once $\EE^H$ has been computed, it is enough to add the equations involving the permutations in a transversal $\{g_1=e,g_2,\ldots, g_{[G:H]}\}$ of $H\setminus G$:
\begin{eqnarray*}
\left .
\begin{array}{l}
  p_{\X} = p_{g_2 \X}\\
 p_{\X} = p_{g_3 \X}\\
  \dots  \\
 p_{\X} = p_{g_{[G:H]} \X}
\end{array}
\right \} \mbox{for all} \; \X\in \B.
\end{eqnarray*}
Notice that the union in (\ref{orb}) is not necessarily disjoint as
it may happen that $\{g_i \X\}_H=\{g_j \X\}_H$ for $i\neq j$. In this
case, the equality $p_{g_j\X}=p_{g_j \X}$ already holds in the space
$\LL^H$ and does not provide any new restriction. In order to avoid this
situation and obtain a minimal set of equations for $\LL^G$,  we
request the special conditions on the $\X\in \B$ in the statement
of the theorem.
%

%

\begin{proof}
For each model $\MM$, we prove that the corresponding equations are linearly independent and there are as many equations as the codimension of $\LL^{\MM}$. By Proposition \ref{dim_models}, the codimension of $\LL^{\MM}$ is $2^{2n-1}$ for $\SSM$, $3\cdot 4^{n-1}$ for $\Ka$, $7\cdot 2^{2n-3}-2^{n-2}$ for $\Kb$, and $4^n-\frac{2^{2n-3}+1}{3}-2^{n-2}$ for $\JC$.
In the sequel, we refer to the groups by the name of the equivariant model associated to them.

\begin{itemize}
\item[$\SSM$:] 
As $\SSM$ is the group $\{e,(\a\t)(\c\g)\}$, a set of
equations for $\SSM$ is  $\{p_{\X}=p_{(\a\t)(\c\g)\X}\}.$ Fixing
$\x_1$ in $\{\a,\c\}$ we obtain $2^{2n-1}$ linearly independent
equations (equations involving different coordinates). The
codimension of $\LL^{\SSM}$ is equal to $2^{2n-1}$, which coincides
with the number of equations given, and thus this set of equations defines $\LL^{\SSM}.$

\item[$\Ka$:] Since a transversal of $\SSM \setminus \Ka$ is $\left \{
e, (\a\c)(\g\t)\right \}$,
the hyperplanes $p_{\X}=p_{(\a\c)(\g\t)\X}$ contain $\LL^{\Ka}$ but
not $\LL^{\SSM}.$ Moreover,  using (\ref{orb}) we see that the orbit $\{\X\}_{\Ka}$ decomposes into the disjoint union of $\{\X\}_{\SSM}$ and $\{(\a\c)(\g\t)\X\}_{\SSM}$ for any $\X \in \B$. Therefore, the equations given for $\Ka$ involve different coordinates than those in $\EE^{\SSM}$. Requiring $\x_1=\a$, we obtain $4^{n-1}$ linearly independent new equations. Thus $\EE^{\Ka}$ defines the space $\LL^{\Ka}$ because the number of linearly independent equations provided, $2^{2n-1}+4^{n-1}=3\cdot 4^{n-1}$, coincides with the codimension of $\LL^{\Ka}$.

\item[$\Kb$:] The set $\left\{e,(\a\g)\right \}$ is a transversal of $\Ka \setminus \Kb$. In order to show that the equations provided are linearly independent to those of $\EE^{\Ka}$, we apply (\ref{orb}) to this transversal to obtain
$\{\X\}_{\Kb}=\{\X\}_{\Ka}\cup \{(\a\g)\X\}_{\Ka}$. If $\X\notin \B_{\a\g \mid \c\t}$, then
$\{(\a\g)\X\}_{\Ka}$ and $\{\X\}_{\Ka}$ are disjoint, so each equation $p_{\X}=p_{(\a\g)\X}$ is linearly independent from $\EE^{\Ka}$.
The set $\B\setminus \B_{\a\g \mid \c\t}$ has cardinal $4^n-2^{n+1}$ and, if $\X \in \B\setminus \B_{\a\g \mid \c\t}$, each orbit $\{\X\}_{\Kb}$ has cardinality 8. Therefore, the number of different orbits for $\X \in \B\setminus \B_{\a\g \mid \c\t}$ is $(4^n-2^{n+1})/8=2^{2n-3}-2^{n-2}.$  Moreover, the choice of $\X$'s in $\B\setminus \B_{\a\g \mid \c\t}$ with $x_1=\a$ and satisfying ``if $\t$ appears in $\X$, there is
a $\c$ in a preceding position'' guarantees that we take only
one element in each $\{\X\}_{\Kb}$, and thus we are adding
exactly one equation for each of these $\X's.$ Overall, there are
$3\cdot 4^{n-1}+(2^{2n-3}-2^{n-2})=7\cdot 2^{2n-3} -2^{n-2}$ linearly
independent equations in $\EE^{\Kb}$. This number coincides with the
codimension of $\LL^{\Kb}$  and these equations define $\LL^{\Kb}$.

\item[$\JC$:] A transversal of $\Kb \setminus \JC$ is $\left \{e,
(\a\c),(\a\t) \right \}$, therefore (\ref{orb}) applies to give $\{\X\}_{\JC}=\{\X\}_{\Kb}\cup \{(\a\c)\X\}_{\Kb}\cup \{(\a\t)\X\}_{\Kb}$.
\begin{itemize}
\renewcommand{\labelitemii}{$\circ$}
\item if $\X\in \B_{\a\c\mid \g\t}\setminus \B_0$, then
  $\{(\a\c)\X\}_{\Kb}=\{\X\}_{\Kb}$ and $\{\X\}_{\JC}$ is the
  disjoint union of $\{\X\}_{\Kb}$ and $\{(\a\t)\X\}_{\Kb}$. As such, each equation $p_{\X}=p_{(\a\t)\X}$ is linearly independent from $\EE^{\Kb}.$

Moreover, if $\X\in \B_{\a\c\mid \g\t}\setminus \B_0$ is of the form $(\a^{l})(\c^{m})\x_{l+m+1}\ldots\x_n,$ we have  $2^{n-1}-1$ such equations and they are linearly independent.

\item if $\X\in \B\setminus \B_2$ then  the three orbits
$\{(\a\c)\X\}_{\Kb}$, $\{(\a\t)\X\}_{\Kb}$, and $\{\X\}_{\Kb}$ have
8 elements each and are disjoint.
Therefore, for these $\X$'s, each equation of type $\{
p_{\X}=p_{(\a\c)\X}\}$ or $\{p_{\X}=p_{(\a\t)\X}\}$ is linearly independent from $\EE^{\Kb}.$ Moreover, as $\B\setminus \B_2$ has cardinal $4^n-3\cdot 2^{n+1}+8$ and is covered by these orbits, we have $\frac{4^n-3\cdot 2^{n+1}+8}{24}=\frac{1}{3}(2^{2n-3}+1)-2^{n-2}$ different orbits. The restriction to the elements of the form $(\a^{l})(\c^{m})\x_{l+m+1}\ldots\x_n$ and satisfying that ``if $\t$ appears in
$\X$, there is some $\g$ in a preceding position'' guarantees that the
equations are written only only once for each orbit.
\end{itemize}
Summing up, there are $$7\cdot 2^{2n-3} -2^{n-2}+\left (
  2^{n-1}-1+2\left(\frac{1}{3}(2^{2n-3}+1)-2^{n-2}\right) \right)$$
linearly independent equations in $\EE^{\JC}$ that contain
$\LL^{\JC}$. As this number is equal to the codimension
$4^n-\frac{2^{2n-3}+1}{3}-2^{n-2}$ of $\LL^{\JC}$, the proof is
complete.
\end{itemize}
All the equalities among orbits used in this proof are summarized in the following table (where $\dots$ means `the set on the left'  and $"$ means `the set on the top').

\begin{center}

 \begin{tabular}{c||c|c|c|c|c}
 & $\{\X\}_{\mathtt{GMM}}$ & $\{\X\}_{\SSM}$ & $\{\X\}_{\Ka}$ & $\{\X\}_{\Kb}$ & $\{\X\}_{\JC}$  \\
 \hline
 $\B_0$ & $\{\X\}$ & $\dots \cup \{(\a\t)(\c\g)\X\}$ & $\dots \cup \{(\a\c)(\g\t)\X\}_{\SSM}$ & $\dots$ & $\dots$ \\
 $\B_{\a\g\mid \c\t}$ & $"$ & $"$ & $"$ & $\dots$ & $\dots \cup \{(\a\c)\X\}_{\Kb}$ \\
 $\B_{\a\c\mid \g\t}$ & $"$ & $"$ & $"$ & $\dots \cup \{(\a\g)\X\}_{\Ka}$ & $\dots \cup \{(\a\t)\X\}_{\Kb}$ \\
 $\B_{\a\t\mid \c\g}$ & $"$ & $"$ & $"$ & $\dots \cup \{(\a\g)\X\}_{\Ka}$ & $\dots \cup \{(\a\c)\X\}_{\Kb}$ \\
 $\B\setminus \B_2$ & $"$ & $"$ & $"$ & $\dots \cup \{(\a\g)\X\}_{\Ka}$ & $\dots \cup \{(\a\c)\X\}_{\Kb}\cup \{(\a\t)\X\}_{\Kb}$
 \end{tabular}
\end{center}

\vspace{3mm}

\end{proof}

\begin{rk}\rm The sets of equations of Theorem
\ref{teo_indep_equations} has been successfully used in
\cite{KDGC} for model selection. Although the dimensions of these
linear spaces are exponential in $n$, in practice it is not
necessary to consider the full set of equations, but
only those containing the patterns observed in the data. This
is crucial for the applicability of the method, since  the number of different columns in an alignment is
really small compared to the dimension of these spaces.
\end{rk}

\begin{ex}
As an example, we compute a minimal system of equations for $\SSM$,
$\Ka$, $\Kb$, and $\JC$ in the case of $3$ leaves.

\noindent \emph{Equations for $\LL^{\SSM}$:} $\mathbb{E}^{\SSM}$ is
composed of the following equations:
\begin{eqnarray*}
p_{\a\a\a } = p_{\t\t\t}, \quad \quad p_{\a\a\c } = p_{\t\t\g},
\quad \quad p_{\a\a\g } = p_{\t\t\c},  \quad \quad
p_{\a\a\t } = p_{\t\t\a}, \\
p_{\a\c\a } = p_{\t\g\t}, \quad \quad p_{\a\c\c } = p_{\t\g\g},
\quad \quad p_{\a\c\g } = p_{\t\g\c}, \quad \quad
p_{\a\c\t } = p_{\t\g\a}, \\
p_{\a\g\a } = p_{\t\c\t}, \quad \quad p_{\a\g\c } = p_{\t\c\g},
\quad \quad p_{\a\g\g } = p_{\t\c\c}, \quad \quad
p_{\a\g\t } = p_{\t\c\a}, \\
p_{\a\t\a } = p_{\t\a\t}, \quad \quad p_{\a\t\c } = p_{\t\a\g},
\quad \quad p_{\a\t\g } = p_{\t\a\c}, \quad \quad
p_{\a\t\t } = p_{\t\a\a}, \\
p_{\c\a\a } = p_{\g\t\t},  \quad \quad p_{\c\a\c } = p_{\g\t\g},
\quad \quad p_{\c\a\g } = p_{\g\t\c}, \quad \quad
p_{\c\a\t } = p_{\g\t\a}, \\
p_{\c\c\a } = p_{\g\g\t},  \quad \quad p_{\c\c\c } = p_{\g\g\g},
\quad \quad p_{\c\c\g } = p_{\g\g\c}, \quad \quad
p_{\c\c\t } = p_{\g\g\a}, \\
p_{\c\g\a } = p_{\g\c\t}, \quad \quad p_{\c\g\c } = p_{\g\c\g},
\quad \quad p_{\c\g\g } = p_{\g\c\c}, \quad \quad
p_{\c\g\t } = p_{\g\c\a}, \\
p_{\c\t\a } = p_{\g\a\t}, \quad \quad p_{\c\t\c } = p_{\g\a\g},
\quad \quad p_{\c\t\g } = p_{\g\a\c}, \quad \quad p_{\c\t\t } =
p_{\g\a\a}.
\end{eqnarray*}

\noindent \emph{Equations for $\LL^{\Ka}$:} $\mathbb{E}^{\Ka}$ is formed by $\mathbb{E}^{\SSM}$ and
\begin{eqnarray*}
p_{\a\a\a } = p_{\c\c\c}, \quad \quad p_{\a\a\c } = p_{\c\c\a},
\quad \quad p_{\a\a\g } = p_{\c\c\t},  \quad \quad
p_{\a\a\t } = p_{\c\c\g}, \\
p_{\a\c\a } = p_{\c\a\c}, \quad \quad p_{\a\c\c } = p_{\c\a\a},
\quad \quad p_{\a\c\g } = p_{\c\a\t}, \quad \quad
p_{\a\c\t } = p_{\c\a\g}, \\
p_{\a\g\a } = p_{\c\t\c}, \quad \quad p_{\a\g\c } = p_{\c\t\a},
\quad \quad p_{\a\g\g } = p_{\c\t\t}, \quad \quad
p_{\a\g\t } = p_{\c\t\g}, \\
p_{\a\t\a } = p_{\c\g\c}, \quad \quad p_{\a\t\c } = p_{\c\g\a},
\quad \quad p_{\a\t\g } = p_{\c\g\t}, \quad \quad p_{\a\t\t } =
p_{\c\g\g}.
\end{eqnarray*}

\noindent \emph{Equations for $\LL^{\Kb}$:} $\mathbb{E}^{\Kb}$ is formed by $\mathbb{E}^{\Ka}$ and
\begin{eqnarray*}
p_{\a\a\g} =p_{\g\a\a},  \quad \quad p_{\a\c\g}=p_{\g\c\a},   \quad
\quad
p_{\a\c\t}=p_{\g\c\t},\\
p_{\a\g\a}=p_{\g\a\g},\quad \quad p_{\a\g\c} =p_{\g\a\c},  \quad
\quad p_{\a\g\g}  = p_{\g\a\a}.
\end{eqnarray*}

\noindent \emph{Equations for $\LL^{\JC}$:} $\mathbb{E}^{\JC}$ is formed by $\mathbb{E}^{\Kb}$ and
\begin{eqnarray*}
p_{\a\a\c} =p_{\t\t\c},  \qquad p_{\a\c\a}=p_{\t\c\t},   \qquad
p_{\a\c\c}=p_{\t\c\c},   \qquad
p_{\a\c\g}=p_{\c\a\g}, \qquad p_{\a\c\g}=p_{\t\c\g}.
\end{eqnarray*}
\end{ex}


\subsection*{Identifiability of phylogenetic mixtures}

In this section we study the identifiability of phylogenetic
mixtures. To this end, we use projective algebraic varieties
and techniques from algebraic geometry. It is not our intention to
give the reader a background on these tools, so we refer to the algebraic geometry book
\cite{Harris} and, more
specifically, to \cite{APRS} for the usage of these techniques in
the study of  phylogenetic mixtures.

There is a natural isomorphism between the points  lying in the
hyperplane $H$ considered above, $H=\{p=(p_{\a \ldots \a}, \dots, p_{\t \ldots \t})
\in \LL : \sum p_{\x_1 \dots \x_n} =1 \}$, and the open affine
subset $\{ p=[p_{\a \ldots \a}: \dots : p_{\t\ldots \t}] :\sum
p_{\x_1 \dots \x_n} \neq 0 \}$ of $\PP^{4^n-1}=\PP(\LL)$. We use the notation
$[p_{\a\ldots \a}:\dots :p_{\t\ldots
\t}]$ for projective coordinates (in contrast to $(p_{\a\ldots \a}, \dots, p_{\t\ldots
\t})$ used for affine coordinates).
The \emph{projective phylogenetic variety} $\PP V_{T}^{\MM}$
associated to a phylogenetic tree $T$ is the  projective closure  in
$\PP(\LL)$ of   the image of the stochastic
parameterization $\phi_T^{\MM}$  defined above. That is, it is the smallest projective variety  in $\PP(\LL)$ containing $\im  \, \phi_T^{\MM}$ via the above isomorphism.

In what follows, we explain the relationship between this new
variety and $CV_T^{\MM}$ and $V_T^{\MM}.$
By Remark \ref{homog_parametr}, it becomes clear
that $CV_{T}^{\MM}$ equals the affine cone over the projective
phylogenetic variety $\mathbb{P}V_{T}^{\MM}$ (for the general Markov model, see also
\cite[Proposition 1]{Allman2004b}).
This implies that $\dim{CV_T^{\MM}}=\dim{\PP
V_{T}^{\MM}}+1$, and if $p=(p_{\texttt{A...A}},
\dots, p_{\texttt{T...T}})$ belongs to $CV_T^{\MM}$, then
$q:=[p_{\texttt{A...A}}: \dots : p_{\texttt{T...T}}]$ belongs to
$\PP V_T^{\MM}$. Moreover, if $\lambda:=\sum p_{\x_1 \dots \x_n}$ is not zero, then $(\frac{p_{\texttt{A...A}}}{\lambda},
\dots , \frac{p_{\texttt{T...T}}}{\lambda})$ is a point in the affine stochastic phylogenetic variety $V_T^{\MM}$.

Before defining identifiability of mixtures, we consider the following construction of projective algebraic varieties.

\begin{defi}
Given two projective varieties $X,Y \subset \PP^m$, the
\emph{join of $X$ and $Y$}, $X \vee Y$, is the smallest variety
in $\PP^m$ containing all lines $\overline{xy}$ with $x \in X$,
$y \in Y$, and $x \neq y$ (see \cite[8.1]{Harris} for details). Similarly, one defines the \emph{join of projective
varieties }$X_1, \dots, X_h \subset \PP^m$, $\vee_{i=1}^h X_i$, as
the smallest subvariety in $\PP^m$ containing all the linear varieties
spanned by $x_1, \dots ,x_h$ with $x_i \in X_i$ and  $x_i \neq
x_j$. It is known that
\begin{eqnarray*}
\dim{(\vee_{i=1}^h X_i)}\leq \min{\{ \sum_{i=1}^h
\dim{(X_i)+h-1},m\}}.
\end{eqnarray*}
The right hand side of this inequality is usually known as the
\emph{expected dimension} of $\vee_{i=1}^h X_i$.
\end{defi}

 For instance, if we consider the join
$\vee_{i=1}^h\PP V_{T_i}^{\MM}$ for certain tree topologies $T_i$
on the leaf set $[n]$ and a given evolutionary model $\MM$, then
there is a (dominant rational) map
\begin{eqnarray}\label{rational_map}
\PP V_{T_1}^{\MM} \times \PP V_{T_2}^{\MM} \times \ldots \times \PP V_{T_h}^{\MM} \times \PP^{h-1} \dashrightarrow \vee_{i=1}^h\PP V_{T_i}^{\MM} \subset \PP(\LL),
\end{eqnarray}
which is the projective closure of the parameterization
$\phi_{T_1} \vee \ldots \vee \phi_{T_h}$ defined by
$$\begin{array}{rcl}
Par_{s\MM}(T_1)  \times \ldots \times Par_{s\MM}(T_h) \times
\Omega & \lra & \mathcal{L}
\\
\left (\xi_1, \ldots, \xi_h, \mathbf{a}  \right) & \mapsto &
\sum_j a_i \phi_{T_i}^{\MM}(\xi_i). \end{array}
$$
Here, $\Omega=\{\mathbf{a}=(a_1, \ldots, a_h) \mid \sum_i a_i=1\}$
is isomorphic to an affine open subset of $\PP^{h-1}$ .
In this setting, an $h$-mixture on $\{T_1, \ldots, T_h\}$ corresponds to
a point in the variety $\vee_{i=1}^h\PP V_{T_i}^{\MM}$.
We will use this
algebraic variety to study the identifiability of phylogenetic
mixtures.

When considering unmixed models $\MM$ on trivalent trees on
$n$ taxa,  \textit{generic identifiability of the tree topology} is
equivalent to the projective varieties $\PP V_{T}^{\MM}$ and  $\PP V_{T'}^{\MM}$ being different when  $T\neq T'$ (see \cite{Allman2008}). The identifiability of
the continuous parameters must take into account the possibility
of permuting the labels of the states at the interior nodes, as
such permutations give  rise to the same joint distribution at
the leaves. In the language  of algebraic geometry,
\textit{generic identifiability of the continuous parameters} of the
model implies that the map  $\phi_{T}^{\MM}$ is generically finite
(i.e. the preimage of a generic point is a finite number of points;
see \cite{Allman2008}). In this case, the fiber dimension
Theorem \cite[Theorem
11.12]{Harris} applies and we have that $\dim \PP V_{T}^{\MM}$ is
equal to  the number of stochastic parameters of the model,  $\dim Par_{s\MM}(T)$.
Therefore, if the continuous parameters are generically identifiable
for the unmixed trees under $\MM$, then the dimension of the variety $\PP
V_T^{\MM}$ is the same for all trivalent tree topologies on $n$
taxa. This dimension is denoted by $s_{\MM}(n)$.

\begin{ex}\rm The tree topologies and the continuous parameters are
generically identifiable for the unmixed equivariant models
\texttt{JC69}, \texttt{K80}, \texttt{K81}, \texttt{SSM}, and
\texttt{GMM} on  trees with any number of leaves (see \cite{Allman2006} and \cite[Corollary 3.9]{CFS3}).
\end{ex}

From now on we only consider trees without nodes of degree 2, so that the number of free stochastic parameters on a phylogenetic tree on $n$ taxa under $\MM$ is $\leq s_{\MM}(n).$

We recall the definition of generic identifiability of the tree
topologies on $h$-mixtures (see \cite{APRS}).

\begin{defi} The \emph{tree topologies} on  $h$-mixtures under $\MM$ are
\emph{generically identifiable} if for any set of trivalent tree
topologies $\{T_1, \ldots, T_h\}$ and a generic choice of $(\xi_1,
\ldots \xi_h,\mathbf{a}) \in Par_{s\MM}(T_1)  \times \ldots \times
Par_{s\MM}(T_h) \times \Omega$, the equality
\[\phi_{T_1} \vee \ldots \vee \phi_{T_h}(\xi_1, \ldots \xi_h,\mathbf{a})=\phi_{T_1'} \vee \ldots \vee \phi_{T_h'}(\xi_1', \ldots \xi_h',\mathbf{a}'),\]
for tree topologies $\{T'_1,\ldots,T'_h\}$ and $(\xi'_1, \ldots \xi'_h,\mathbf{a}')\in Par_{s\MM}(T'_1)  \times \ldots \times
Par_{s\MM}(T'_h) \times \Omega$ implies \[\{T_1,
\ldots, T_h \}=\{T_1,' \ldots, T_h'\}.\]In terms of algebraic
varieties this is equivalent to saying that the variety
$\vee_{i=1}^h\PP V_{T_i}^{\MM}$ is not contained in
$\vee_{i=1}^h\PP V_{T_i'}^{\MM}$ and vice versa.

\end{defi}

The tree topologies are the discrete parameters of $h$-mixtures.
When considering the  continuous parameters of
$h$-mixtures, the above mentioned label-swapping can  be
disregarded. We give the following
definition according to \cite{RhodesSullivant}.

\begin{defi}  The \emph{continuous parameters} of $h$-mixtures on $T_1,
\dots, T_h$ under an evolutionary model $\MM$ are
\emph{generically identifiable} if, for a generic choice of
stochastic parameters $(\xi_1, \dots ,\xi_h,\mathbf{a})$, the
equality
$$\phi_{T_1} \vee \ldots \vee \phi_{T_h}(\xi_1, \ldots \xi_h,
\mathbf{a})=\phi_{T_1} \vee \ldots \vee \phi_{T_h}(\xi_1', \ldots
\xi_h', \mathbf{a}')$$ for stochastic parameters $(\xi'_1, \ldots
,\xi'_h,\mathbf{a}')$ implies
that there is a permutation $\sigma \in \mathfrak{S}_h$ such that
${\sigma \cdot (T_1,\dots,T_h)}=(T_1,\dots, T_h)$,
$\xi_i'=\xi_{\sigma(i)}$, and $a_i'=a_{\sigma(i)}$ for $i=1,\dots
r$. In other words, we only allow swapping of the continuous parameters when at least two tree topologies coincide.

%

\end{defi}

\begin{defi}
 An $h$-mixture under a model $\MM$
 is said to be \emph{identifiable} if both its tree topologies
 and its  continuous parameters are generically identifiable.
\end{defi}

In terms of algebraic varieties, generic identifiability of
continuous parameters on $h$-mixtures implies that the generic fibers (i.e. preimages of generic points) of the map $\phi_{T_1} \vee \ldots
\vee \phi_{T_h}$ are finite. In this case, the fiber dimension
theorem applied to (\ref{rational_map}) (cf. \cite[Theorem 11.12]{Harris}) gives
\begin{eqnarray*}
\dim{(\vee_{i=1}^h \PP V_{T_i})}=\sum_{i=1}^h
\dim{(\PP V_{T_i})+h-1}.
\end{eqnarray*}

The following result demonstrates the need for careful
  inspection of identifiability of mixtures with many components
  (i.e. large values of $h$).

\begin{thm}\label{noident}
Let $[n]$ be a set of taxa and $\MM$ be an evolutionary model for
which the  continuous parameters
are generically identifiable on trivalent (unmixed) trees. In
addition,  let
$s_{\MM}(n)$ be the
dimension of $\PP V_T^{\MM}$ for any trivalent tree $T$, and set $h_0(n):=\frac{\dim{\DD_{\MM}}}{s_{\MM}(n)+1}$. Then
the $h$-mixtures of trees on $[n]$ evolving under $\MM$ are not identifiable for $h \geq h_0(n)$.
\end{thm}

\begin{rk}\rm
 Note that, in the above definition of $h_0(n)$, $\dim \DD_{\MM}$ also depends on $n$.
\end{rk}

\begin{cor}\label{cormodels}
Let $[n]$ be a set of taxa and $\MM$ be one of the equivariant
models \texttt{JC69}, \texttt{K80}, \texttt{K81}, \texttt{SSM}, or
\texttt{GMM}. Then the   phylogenetic $h$-mixtures
under these models are not identifiable for $h\geq h_0(n)$,  where
\begin{eqnarray*}
h_0(n)=\left \{ \begin{array}{lc}
  \frac{4^{n}}{12(2n-3)+4}, \mbox{ if } \MM=\texttt{GMM}, \\ \\
\frac{2^{2n-1}}{6(2n-3)+2}, \mbox{ if } \MM=\texttt{SSM}, \\ \\
\frac{4^{n-1}}{3(2n-3)+1}, \mbox{ if } \MM=\texttt{K81},  \\ \\
\frac{2^{2n-3}+2^{n-2}}{2(2n-3)+1}, \mbox{ if } \MM=\texttt{K80}, \\ \\
\frac{2^{2n-3}+3\cdot 2^{n-2}+1}{3(2n-2)}, \mbox{ if } \MM=\texttt{JC69}.
\end{array}
\right.
\end{eqnarray*}
\end{cor}

\begin{proof}
Theorem \ref{caractLG} shows that $\LL^{\MM}=\DD_{\MM}$ and
Proposition \ref{dim_models} gives the dimension of $\LL^{\MM}$ in each case. Next, we calculate: $s_{\texttt{GMM}}(n)=12(2n-3)+3$,
$s_{\texttt{SSM}}(n)=6(2n-3)+1$, $s_{\texttt{K81}}(n)=3(2n-3)$,
$s_{\texttt{K80}}(n)=2(2n-3)$, and $s_{\texttt{JC69}}(n)=2n-3$. Applying Theorem \ref{noident}, we conclude the proof.
\end{proof}

\begin{ex}
\rm Consider the Kimura 3-parameter model \texttt{K81} on $n=4$ taxa. For  any $h \geq 4$,
phylogenetic $h$-mixtures are not identifiable by Corollary
\ref{cormodels}. We are not aware of any result proving that
mixtures of 2 or 3 different tree topologies  under this model are
identifiable (either for the  tree parameters or for the  continuous
parameters).
\end{ex}

\begin{ex}
\rm Consider  the Jukes-Cantor model \texttt{JC69} on $n=4$
taxa. Then  Corollary \ref{cormodels} tells us that for $h \geq 3$,
$h$-mixtures are not identifiable. Therefore, for this particular
model on four taxa the cases in which the identifiability holds are known: the tree and
the continuous parameters are generically identifiable for the unmixed
model; the tree parameters are generically identifiable for
2-mixtures \cite[Theorem 10]{APRS}; the continuous parameters are
generically identifiable for 2-mixtures on different tree
topologies and not identifiable for the same tree topology
\cite[Theorem 23]{APRS}; neither the continuous parameters nor the
tree topologies are generically identifiable for mixtures with more than two components (Corollary \ref{cormodels}).
\end{ex}

\begin{proofof}{Theorem \ref{noident}}
Let $edim(h):=h s_{\MM}(n)+h-1$. Then  the variety $\vee_{i=1}^h \PP
V_{T_i}$ has dimension $\leq edim(h)$. Indeed, as $\vee_i
\phi_{T_i}$ is a parameterization of an open subset of
$\vee_{i=1}^h \PP V_{T_i}$, then the dimension of $\vee_{i=1}^h
\PP V_{T_i}$ is less than or equal to $\sum \dim{\PP V_{T_i}}+h-1$.
Moreover, the dimension of $\PP V_{T_i}$ is equal to $s_{\MM}(n)$ if
$T_i$ is trivalent (since the continuous parameters for the
unmixed models under consideration  are generically identifiable)
and is less than $s_{\MM}(n)$  for non-trivalent
trees. Therefore,
$\dim(\vee_{i=1}^h \PP V_{T_i}) \leq edim(h).$

If all $T_i$ are trivalent trees, 
then $\sum \dim{\PP V_{T_i}}+h-1=edim(h)$ and, therefore,
$\dim(\vee_{i=1}^h \PP V_{T_i}) < edim(h)$ if and only if
$\dim(\vee_{i=1}^h \PP V_{T_i})<\sum \dim{\PP V_{T_i}}+h-1$.
Moreover, by the fiber dimension theorem applied to $\vee \phi_{T_i}$,
the equality of dimensions holds if and only if the generic fiber of $\vee
\phi_{T_i}$ has dimension 0. In particular, if $\dim(\vee_{i=1}^h
\PP V_{T_i}) < edim(h)$,  then the continuous parameters of this
phylogenetic mixture are not identifiable.

As $h_0(n)=\frac{\dim{\DD_{\MM}}}{s_{\MM}(n)+1}$, we have that $edim(h_0(n))=h_0(n)
(s_{\MM}(n)+1)-1=\dim{\DD_{\MM}}-1$. Now,  we fix an $h \in \mathbb{N}$
with $h \geq h_0(n)$, so that  $edim(h) \geq
\dim(\DD_{\MM})-1$.

There are two possible scenarios:

(a) For any set of tree topologies $\{T_1, \dots, T_h\}$, the dimension of
$\vee_{i=1}^h \PP V_{T_i}$ is less than $\dim(\DD_{\MM})-1$.

(b) There exists a set of tree topologies $\{T_1, \dots, T_h\}$
for which $\dim(\vee_{i=1}^h \PP V_{T_i}) = \dim(\DD_{\MM})-1.$

Case (a) implies that the dimension of $\vee_{i=1}^h \PP V_{T_i}$ is less than
$edim(h)$ for any set of trivalent tree topologies
$\{T_1, \dots, T_h\}$. Based on the conclusions drawn above, this implies that the
continuous parameters are not generically identifiable.

In case (b),  $\vee_{i=1}^h \PP V_{T_i}$ coincides with  $\PP(\DD_{\MM})$. Indeed, $\vee_{i=1}^h \PP V_{T_i}$ is contained in
$\PP(\DD_{\MM})$, both varieties are irreducible,  and $\dim(\vee_{i=1}^h \PP V_{T_i}) =
\dim(\DD_{\MM})-1=\dim(\PP(\DD_{\MM}))$,  which implies that both
varieties coincide. In particular,  any
$h$-mixture (which is a point in $\PP(\DD_{\MM})$) would be
contained in $\vee_{i=1}^h \PP V_{T_i}$,  and therefore the
topologies are not generically identifiable.
\end{proofof}

\begin{rk}
\rm The negative result of Theorem \ref{noident} should be
complemented with the following positive result of  Rhodes and
Sullivant in \cite{RhodesSullivant}: if $\MM=$ \texttt{GMM} and
one restricts to $h$-mixtures on the same trivalent tree topology
$T$, then the tree topology and the continuous parameters are
generically identifiable if $h < 4^{\lceil\frac{n}{4}\rceil -1}.$
\end{rk}

\bigskip

\section*{Author's contributions}
    All authors contributed equally and the author names order is alphabetical.

\section*{Acknowledgements}
   All authors are partially supported by Generalitat de Catalunya, 2009 SGR 1284. Research of the first and second authors partially supported by Ministerio de Educaci\'on y Ciencia MTM2009-14163-C02-02. We would like to thank both referees for very useful comments that led to major improvements.




\begin{thebibliography}{KDGC12}

\bibitem[APRS11]{APRS}
E~S Allman, S~Petrovic, J~A Rhodes, and S~Sullivant.
\newblock {Identifiability of two-tree mixtures for group-based models}.
\newblock {\em IEEE/ACM Transactions in Computational Biology and
  Bioinformatics}, 8:710--720, 2011.

\bibitem[AR03]{Allman2003}
ES~Allman and JA~Rhodes.
\newblock Phylogenetic invariants for the general {M}arkov model of sequence
  mutation.
\newblock {\em Mathematical Biosciences}, 186(2):113--144, 2003.

\bibitem[AR04]{Allman2004}
ES~Allman and JA~Rhodes.
\newblock Quartets and parameter recovery for the general {M}arkov model of
  sequence mutation.
\newblock {\em Applied Mathematics Research Express}, 2004(4):107--131, 2004.

\bibitem[AR06a]{Allman2006}
ES~Allman and JA~Rhodes.
\newblock The identifiability of tree topology for phylogenetic models,
  including covarion and mixture models.
\newblock {\em Journal of Computational Biology}, 13:1101--1113, 2006.

\bibitem[AR06b]{Allman2006b}
ES~Allman and JA~Rhodes.
\newblock Phylogenetic invariants for stationary base composition.
\newblock {\em Journal of Symbolic Computation}, 41(2):138 -- 150, 2006.

\bibitem[AR08a]{Allman2008}
ES~Allman and JA~Rhodes.
\newblock Identifying evolutionary trees and substitution parameters for the
  general {M}arkov model with invariable sites.
\newblock {\em Mathematical Biosciences}, 211:18--33, 2008.

\bibitem[AR08b]{Allman2004b}
ES~Allman and JA~Rhodes.
\newblock Phylogenetic ideals and varieties for the general {M}arkov model.
\newblock {\em Advances in Applied Mathematics}, 40:127--148, 2008.

\bibitem[CF87]{Cavender87}
J~Cavender and J~Felsenstein.
\newblock Invariants of phylogenies in a simple case with discrete states.
\newblock {\em Journal of Classification}, 4:57--71, 1987.

\bibitem[CFS11]{CFS3}
M~Casanellas and J~Fernandez-Sanchez.
\newblock Relevant phylogenetic invariants of evolutionary models.
\newblock {\em Journal de Mathématiques Pures et Appliquées}, 96:207--229,
  2011.

\bibitem[CH11]{ChaiHousworth}
J~Chai and E~A Housworth.
\newblock On {R}ogers's proof of identifiability for the {GTR} + {G}amma + {I}
  model.
\newblock {\em Systematic Biology}, 2011.

\bibitem[Cha96]{Chang96}
J~T Chang.
\newblock Full reconstruction of {M}arkov models on evolutionary trees:
  identifiability and consistency.
\newblock {\em Mathematical Biosciences}, 137(1):51--73, 1996.

\bibitem[CS05]{CS}
M~Casanellas and S~Sullivant.
\newblock The strand symmetric model.
\newblock In L.~Pachter and B.~Sturmfels, editors, {\em Algebraic Statistics
  for computational biology}, chapter~16. Cambridge University Press, 2005.

\bibitem[DK09]{Draisma}
J~Draisma and J~Kuttler.
\newblock On the ideals of equivariants tree models.
\newblock {\em Mathematische Annalen}, 344:619--644, 2009.

\bibitem[Fel03]{Felsenstein2003}
J~Felsenstein.
\newblock {\em Inferring Phylogenies}.
\newblock Sinauer Associates, Inc., 2003.

\bibitem[FL92]{FuLi}
Y-X Fu and WH~Li.
\newblock Construction of linear invariants in phylogenetic inference.
\newblock {\em Mathematical Biosciences}, 109(2):201 -- 228, 1992.

\bibitem[Har92]{Harris}
J~Harris.
\newblock {\em Algebraic geometry. A first course}, volume 133 of {\em Graduate
  Texts in Mathematics}.
\newblock Springer-Verlag, New York, 1992.

\bibitem[JC69]{JC69}
T~H Jukes and C~R Cantor.
\newblock {Evolution of protein molecules}.
\newblock In H~N Munro, editor, {\em Mammalian Protein Metabolism}, volume~3,
  pages 21--132. Academic Press, 1969.

\bibitem[KDGC12]{KDGC}
A~Kedzierska, M~Drton, R~Guig\'{o}, and M~Casanellas.
\newblock \texttt{SPIn}: model selection for phylogenetic mixtures via linear
  invariants.
\newblock {\em Molecular Biology and Evolution}, 29:929--937, 2012.

\bibitem[Kim80]{Kimura1980}
M~Kimura.
\newblock A simple method for estimating evolutionary rates of base
  substitution through comparative studies of nucleotide sequences.
\newblock {\em Journal of Molecular Evolution}, 16:111--120, 1980.

\bibitem[Kim81]{Kimura1981}
M~Kimura.
\newblock Estimation of evolutionary sequences between homologous nucleotide
  sequences.
\newblock {\em Proceedings of the National Academy of Sciences}, 78:454--458,
  1981.

\bibitem[Lak87]{Lake1987}
JA~Lake.
\newblock A rate-independent technique for analysis of nucleaic acid sequences:
  evolutionary parsimony.
\newblock {\em Molecular Biology and Evolution}, 4:167--191, 1987.

\bibitem[MMS08]{MatsenMosselSteel}
FA~Matsen, E~Mossen, and M~Steel.
\newblock Mixed-up trees: The structure of phylogenetic mixtures.
\newblock {\em Bulletin of Mathematical Biology}, 70:1115--1139, 2008.

\bibitem[Pos01]{Posada2001}
D~Posada.
\newblock The effect of branch length variation on the selection of models of
  molecular evolution.
\newblock {\em Journal of Molecular Evolution}, 52:434--444, 2001.

\bibitem[RS12]{RhodesSullivant}
JA~Rhodes and S~Sullivant.
\newblock Identifiability of large phylogenetic mixture models.
\newblock {\em Bulletin of Mathematical Biology}, 74:212--231, 2012.

\bibitem[Ser77]{Serre}
JP~Serre.
\newblock {\em Linear representations of finite groups}.
\newblock Springer-Verlag, New York, 1977.
\newblock Translated from the second French edition by Leonard L. Scott,
  Graduate Texts in Mathematics, Vol. 42.

\bibitem[SFSJ12]{LMM}
J~Sumner, J~Fernández-Sánchez, and P~Jarvis.
\newblock On lie markov models.
\newblock {\em Journal of Theoretical Biology}, 298:16--31, 2012.

\bibitem[SHP98]{Steel98}
MA~Steel, MD~Hendy, and D~Penny.
\newblock Reconstructing phylogenies from nucleotide pattern probabilities: a
  survey and some new results.
\newblock {\em Discrete Applied Mathematics}, 88(1-3):367--396, 1998.

\bibitem[SHSE92]{Steel1992}
MA~Steel, MD~Hendy, LA~Sz{\'e}kely, and PL~Erd{\H{o}}s.
\newblock Spectral analysis and a closest tree method for genetic sequences.
\newblock {\em Applied Mathematics Letters. An International Journal of Rapid
  Publication}, 5:63--67, 1992.

\bibitem[SS03]{Semple2003}
C~Semple and M~Steel.
\newblock {\em Phylogenetics}, volume~24 of {\em {\rm Oxford Lecture Series in
  Mathematics and its Applications}}.
\newblock Oxford University Press, Oxford, 2003.

\bibitem[SS05]{Sturmfels2005}
B~Sturmfels and S~Sullivant.
\newblock Toric ideals of phylogenetic invariants.
\newblock {\em Journal of Computational Biology}, 12:204--228, 2005.

\bibitem[SV07]{Stefanovic}
D~Stefanovic and E~Vigoda.
\newblock Phylogeny of mixture models: Robustness of maximum likelihood and
  non-identifiable distributions.
\newblock {\em Journal of Computational Biology}, 14:156--189, 2007.

\end{thebibliography}
\end{document}